\documentclass[runningheads, article]{llncs}
\usepackage{amsfonts}
\usepackage{amscd}
\usepackage{amssymb}
\usepackage{amsmath}
\usepackage{bm}
\usepackage{mathrsfs}
\usepackage{mathtools}
\usepackage{here}
\usepackage{subfigure}
\usepackage{examplep}
\usepackage{mdframed}
\usepackage{graphicx}
\usepackage[T1]{fontenc}
\usepackage{color}

\begin{document}

\title{Robust portfolio optimization for recommender systems considering uncertainty of estimated statistics}

\titlerunning{Robust portfolio optimization for recommender systems}

\author{Tomoya Yanagi\inst{1}
\and
Shunnosuke Ikeda\inst{2}
\orcidID{0009-0004-4283-0819} \and \\
Yuichi Takano\inst{3}
\orcidID{0000-0002-8919-1282}
}
\authorrunning{T. Yanagi et al.}

\institute{Graduate School of Science and Technology, University of Tsukuba, \\ 1--1--1 Tennodai, Tsukuba-shi, 305--8573, Ibaraki, Japan
\email{yanagi.tomoya.ta@alumni.tsukuba.ac.jp}\\
\and
Graduate School of Science and Technology, University of Tsukuba, \\ 1--1--1 Tennodai, Tsukuba-shi, 305--8573, Ibaraki, Japan \\
\email{ikeda.shunnosuke@gmail.com}\\
\and
Institute of Systems and Information Engineering, University of Tsukuba, \\ 1--1--1 Tennodai, Tsukuba-shi, 305--8573, Ibaraki, Japan \\
\email{ytakano@sk.tsukuba.ac.jp}}
\maketitle 
\begin{abstract}
This paper is concerned with portfolio optimization models for creating high-quality lists of recommended items to balance the accuracy and diversity of recommendations.
However, the statistics (i.e., expectation and covariance of ratings) required for mean--variance portfolio optimization are subject to inevitable estimation errors.
To remedy this situation, we focus on robust optimization techniques that derive reliable solutions to uncertain optimization problems. 
Specifically, we propose a robust portfolio optimization model that copes with the uncertainty of estimated statistics based on the cardinality-based uncertainty sets. 
This robust portfolio optimization model can be reduced to a mixed-integer linear optimization problem, which can be solved exactly using mathematical optimization solvers. 
Experimental results using two publicly available rating datasets demonstrate that our method can improve not only the recommendation accuracy but also the diversity of recommendations compared with conventional mean--variance portfolio optimization models. 
Notably, our method has the potential to improve the recommendation quality of various rating prediction algorithms.

\keywords{Recommender system \and Portfolio optimization \and Robust optimization \and Diversity.}
\end{abstract}

\section{Introduction}\label{sec:intro}
\subsection{Background}\label{subsec:background}
Recent advances in information and communication technologies have enabled consumers to browse and purchase a wide variety of items through online services.
Meanwhile, consumers often struggle to find items that match their preferences among a plethora of options. 
To overcome this information overload problem, recommender systems~\cite{aggarwal2016recommender} 
have been widely implemented in online services. 
These systems generate personalized lists of unknown but potentially preferred items.
In fact, integrating these systems into online platforms simplifies the search process and improves user experience~\cite{lu2015recommender}. 

Much of the research to date has focused on improving the prediction accuracy of recommender systems~\cite{bobadilla2013recommender,iwanaga2019improving}.
To this end, a variety of rating prediction algorithms have been proposed, ranging from traditional methods such as collaborative filtering~\cite{su2009survey} and matrix factorization~\cite{koren2009matrix} to recent deep learning techniques~\cite{gao2022graph,zhang2019deep}. 

In contrast, it has been pointed out that the most accurate recommendations do not always satisfy users~\cite{mcnee2006being}. 
Specifically, favorable recommendations must take into account the diversity and novelty of items provided for users~\cite{castells2021novelty,kaminskas2016diversity}.
However, most recommender systems tend to offer popular items, leading to reduced diversity.
Increasing the diversity and novelty of recommendations remains a major challenge in the advancement of recommender systems.

\subsection{Related work}\label{subsec:relatedwork}

Prior studies have developed various algorithms to improve recommendation diversity~\cite{castells2021novelty,kunaver2017diversity,moller2020not}. 
In particular, we explore applying financial portfolio theory~\cite{elton2009modern} to recommender systems, aiming to balance the accuracy and diversity of recommendations.

Mean--variance portfolio analysis has been applied to ranking systems in information retrieval and collaborative filtering~\cite{wang2009mean,wang2009portfolio}.
Shi et al.~\cite{shi2012adaptive} improved this ranking strategy by using the latent factor model for user ratings. 
Kwon~\cite{kwon2008improving} focused on the variance of ratings to increase the precision of top-$N$ recommendations. 
Zhang and Hurley~\cite{zhang2008avoiding} and Hurley and Zhang~\cite{hurley2011novelty} formulated several optimization models for selecting a list of diverse items. 
Xiao et al.~\cite{xiao2020recrisk} used a portfolio optimization model to recommend relevant services based on service risk facets.
Yasumoto and Takano~\cite{yasumotomean} developed shrinkage estimation methods to improve the accuracy of estimating the rating covariance matrix. 

These studies have shown that the quality of recommendation lists can be improved by applying portfolio optimization. 
However, the statistics (i.e., expectation and covariance of ratings) required for mean--variance portfolio optimization are subject to inevitable estimation errors. 
In fact, it has been pointed out that the mean--variance portfolio optimization model often performs badly because of such estimation errors~\cite{broadie1993computing}. 

To remedy this situation, we focus on robust optimization techniques~\cite{bertsimas2011theory,kim2018recent} that derive reliable solutions to uncertain optimization problems. 
To the best of our knowledge, none of the prior studies have considered recommender systems that incorporate uncertainty of estimated statistics to select high-quality lists of recommendations.

\subsection{Our contribution}\label{subsec:contribution}
The goal of this paper is to improve the recommendation quality of portfolio optimization models by using robust optimization techniques. 
Specifically, we propose a robust portfolio optimization model that copes with the uncertainty of estimated statistics based on the cardinality-based uncertainty sets~\cite{bertsimas2004price}. 
Our robust portfolio optimization model can be reduced to a mixed-integer linear optimization problem, which can be solved exactly using mathematical optimization solvers. 

To verify the effectiveness of our method for recommendation, we conducted numerical experiments using two publicly available rating datasets. 
Experimental results demonstrate that our method can improve not only the recommendation accuracy but also the diversity of recommendations compared with conventional mean--variance portfolio optimization models. 
This emphasizes that by properly handling the uncertainty in the expectation and covariance of user ratings, we can generate high-quality lists of recommendations that fulfill users' diverse requirements. 
Moreover, our method has the potential to improve the
recommendation quality of various rating prediction algorithms. 

\section{Mean--variance portfolio optimization model}\label{sec:meanval}
In this section, we outline the mean--variance portfolio optimization model~\cite{yasumotomean}, which is a combination of mean--variance portfolio analysis~\cite{wang2009mean,wang2009portfolio} and the cardinality-constrained optimization model~\cite{hurley2011novelty,zhang2008avoiding} for recommendation. 

Let $U$ and $I$ denote the sets of users and items, respectively. 
The rating matrix is then defined as
\[
    \bm{R} \coloneqq (r_{ui})_{(u,i) \in U \times I} \in \mathbb{R}^{\lvert U \rvert \times \lvert I \rvert},
\]
where $r_{ui}$ denotes the rating indicated by user $u \in U$ to item $i \in I$. 
Note that this matrix contains many missing entries because each user rates only a small subset of items. 
We denote by $Q \subseteq U \times I$ the subset of user--item pairs for which ratings have been observed. 
Recommender systems involve predicting unknown ratings for user--item pairs $(u,i) \not\in Q$ and then suggesting items that are likely to be favored by each user.

Let $\bm{x} \coloneqq (x_{i})_{i \in I} \in \{0,1\}^{\lvert I \rvert}$ be a vector composed of binary design variables for selecting items to be recommended, where $x_{i} = 1$ indicates that item $i \in I$ is recommended for a target user. 
We define $\hat{{\mu}}_{ui}$ as the expected rating of user $u \in U$ for item $i \in I$; this value can be estimated using rating prediction algorithms (e.g., collaborative filtering \cite{su2009survey}, matrix factorization \cite{koren2009matrix}, and deep learning techniques~\cite{gao2022graph,zhang2019deep}). 
We also define $\hat{\sigma}_{ij}$ as the rating covariance for a pair $(i,j) \in I \times I$ of items; this value can be calculated from users' ratings observed for both items $i,j \in I$. 
Shrinkage estimation methods were also developed to improve the accuracy of estimating the rating covariance matrix~\cite{yasumotomean}. 

For each user $u \in U$, we denote by $R_u(\bm{x})$ a total rating function, which is the sum of ratings for recommended items. 
The expectation and variance of the total rating can be estimated as follows:
\begin{equation}
    \text{E}[R_u(\bm{x})] \approx \sum_{i \in I} \hat{{\mu}}_{ui} x_{i}, \quad 
    \text{Var}[R_u(\bm{x})] \approx \sum_{i \in I}\sum_{j \in I} \hat{\sigma}_{ij} x_{i} x_{j}. \label{eq:meanVar}
\end{equation}

Let $I_u \subseteq \{i \in I \mid (u,i) \notin Q\}$ be a set of candidate items to be recommended for user $u \in U$. 
For each target user $u \in U$, the mean--variance portfolio optimization model for selecting $N$ recommended items is formulated as the following binary optimization problem: 
\begin{align} 
    \underset{\bm{x}}{\text{minimize}} \quad & \alpha \sum_{i \in I_u} \sum_{j \in I_u} \hat{\sigma}_{i j} x_{i} x_{j} - (1 -\alpha) \sum_{i \in I_u} \hat{\mu}_{ui} x_{i} \label{eq:portfolio_obj} \\ 
    \text {subject to} \quad & \sum_{i \in I_u} x_{i} = N, \label{eq:portfolio_N} \\
    & x_{i} \in\{0,1\} \quad\left(i \in I_u\right), \label{eq:portfolio_x}
\end{align}
where $\alpha \in [0,1]$ is a hyperparameter for adjusting the mean--variance trade-off. 
The objective function (Eq.~\eqref{eq:portfolio_obj}) is the weighted sum of two objectives (Eq.~\eqref{eq:meanVar}).
Eq.~\eqref{eq:portfolio_N} specifies the number of recommended items, and Eq.~\eqref{eq:portfolio_x} imposes the binary constraint on each entry of $\bm{x}$ for selecting recommendations.

\section{Robust portfolio optimization model}\label{sec:robust-opt}
In this section, we propose a robust portfolio optimization model for recommendation by incorporating the cardinality-based uncertainty sets~\cite{bertsimas2004price} into the mean--variance model~(Eqs.~\eqref{eq:portfolio_obj}--\eqref{eq:portfolio_x}). 

Specifically, for each user $u \in U$ we consider the following ranges of variation in the expectation and covariance of ratings:
\begin{align}
\mu_{ui} & \in [\hat{\mu}_{ui} - \delta_{ui}^{(\mu)},~\hat{\mu}_{ui} + \delta_{ui}^{(\mu)}] \quad (i \in I), \notag \\
\sigma_{ij} & \in [\hat{\sigma}_{ij} - \delta_{ij}^{(\sigma)},~\hat{\sigma}_{ij} + \delta_{ij}^{(\sigma)}] \quad ((i,j) \in I \times I), \notag
\end{align}
where $\delta_{ui}^{(\mu)} \ge 0$ and $\delta_{ij}^{(\sigma)} \ge 0$ are the magnitudes of variation in the expectation and covariance, respectively. 

Let $S^{(\mu)} \subseteq I$ and $S^{(\sigma)} \subseteq I \times I$ denote subsets of items and item pairs that are accompanied by variation in the expectation and covariance, respectively. 
For each target user $u \in U$, we consider minimizing the objective function (Eq.~\eqref{eq:portfolio_obj}) in the worst case as follows: 
\begin{align} 
    \underset{\bm{x}}{\text{minimize}} \quad & \alpha \left( \sum_{i \in I_u} \sum_{j \in I_u} \hat{\sigma}_{ij} x_{i} x_{j} + \max_{\substack{S^{(\sigma)} \subseteq I_u \times I_u \\ \lvert S^{(\sigma)} \rvert \le \Gamma^{(\sigma)}}} \left\{ \sum_{(i,j) \in S^{(\sigma)}} \hspace*{-4mm} \delta_{ij}^{(\sigma)} x_{i} x_{j} \right\}\right) \notag \\
    & - (1 -\alpha) \left( \sum_{i \in I_u} \hat{\mu}_{ui} x_{i} - \max_{\substack{S^{(\mu)} \subseteq I_u \\ \lvert S^{(\mu)} \rvert \le \Gamma^{(\mu)}}} \left\{\sum_{i \in S^{(\mu)}} \delta_{ui}^{(\mu)} x_{i} \right\} \right) \label{obj:robust1} \\ 
    \text {subject to} \quad 
    & \mbox{Eqs.~\eqref{eq:portfolio_N} and \eqref{eq:portfolio_x}}, \label{con1:robust1}
\end{align}
where $\Gamma^{(\mu)}, \Gamma^{(\sigma)} \in \mathbb{Z}_{+}$ are integer-valued parameters for limiting the cardinalities of $S^{(\mu)}$ and $S^{(\sigma)}$. 

We are now in a position to derive our formulation of the robust portfolio optimization model for recommendation. 

\begin{theorem} \label{thm:robust_portfolio} \rm
    For each target user $u \in U$, problem (Eqs.~\eqref{obj:robust1}--\eqref{con1:robust1}) can be reformulated as the following mixed-integer optimization problem: 
\begin{align}
    \underset{\bm{p}, \bm{q}, \bm{x}, y, z}{\mathrm{minimize}} \quad 
    & \alpha \left(\sum_{i \in I_u} \sum_{j \in I_u} \hat{\sigma}_{ij} x_{i} x_{j} + z \Gamma^{(\sigma)} + \sum_{i \in I_u} \sum_{j \in I_u} q_{ij} \right) \notag \\
    & - (1-\alpha) \left(\sum_{i \in I_u} \hat{\mu}_{ui} x_{i} - y \Gamma^{(\mu)} - \sum_{i \in I_u} p_{i} \right) \label{obj:robust2}\\  
    \mathrm{subject~to} \quad
    & \mathrm{Eqs}.~\eqref{eq:portfolio_N}~\mathrm{and}~\eqref{eq:portfolio_x}, \label{con1:robust2}\\ 
    & \delta^{(\mu)}_{ui} x_{i} \leq y + p_{i}, \quad p_{i} \geq 0 \quad (i \in I_u), \label{con2:robust2}\\
    & \delta^{(\sigma)}_{ij} x_{i} x_{j} \leq z + q_{ij}, \quad q_{ij} \geq 0 \quad ((i,j) \in I_u \times I_u), \label{con3:robust2}\\
    & y \geq 0, \quad z \geq 0, \label{con4:robust2}
\end{align}
where $\bm{p} \coloneqq (p_i)_{i \in I} \in \mathbb{R}^{\lvert I \rvert}$, $\bm{q} \coloneqq (q_{ij})_{(i,j) \in I \times I} \in \mathbb{R}^{\lvert I \times I \rvert}$, $y \in \mathbb{R}$, and $z \in \mathbb{R}$ are dual design variables. 
\end{theorem}

\begin{proof}
The latter inner maximization problem in Eq.~\eqref{obj:robust1} can be rewritten as 
\begin{align}
    \underset{\bm{\gamma}^{(\mu)}}{\text{maximize}} \quad & \sum_{i \in I_u} \delta_{ui}^{(\mu)} x_i \gamma_i^{(\mu)} \label{obj:inner_min1} \\
    \text{subject to} \quad 
    & \sum_{i \in I_u} \gamma_i^{(\mu)} \le \Gamma^{(\mu)}, \label{con1:inner_min1} \\
    & 0 \le \gamma_i^{(\mu)} \le 1 \quad (i \in I_u), \label{con2:inner_min1}
\end{align}
where $\bm{\gamma}^{(\mu)} \coloneqq (\gamma_{i}^{(\mu)})_{i \in I} \in \mathbb{R}^{\lvert I \rvert}$ is an auxiliary design variable serving as the subset $S^{(\mu)}$. 
Similarly, the former inner maximization problem in Eq.~\eqref{obj:robust1} can be rewritten as 
\begin{align}
    \underset{\bm{\gamma}^{(\sigma)}}{\text{maximize}} \quad & \sum_{i \in I_u} \sum_{j \in I_u} \delta_{ij}^{(\sigma)} x_{i} x_{j} \gamma_{ij}^{(\sigma)} \label{obj:inner_min2} \\
    \text{subject to} \quad 
    & \sum_{i \in I_u} \sum_{j \in I_u} \gamma_{ij}^{(\sigma)} \le \Gamma^{(\sigma)}, \label{con1:inner_min2} \\
    & 0 \le \gamma_{ij}^{(\sigma)} \le 1 \quad ((i,j) \in I_u \times I_u), \label{con2:inner_min2}
\end{align}
where $\bm{\gamma}^{(\sigma)} \coloneqq (\gamma_{ij}^{(\sigma)})_{(i,j) \in I \times I} \in \mathbb{R}^{\lvert I \times I \rvert}$ is an auxiliary design variable serving as the subset $S^{(\sigma)}$. 

The proof is then completed by transforming these maximization problems (Eqs.~\eqref{obj:inner_min1}--\eqref{con2:inner_min1} and Eqs.~\eqref{obj:inner_min2}--\eqref{con2:inner_min2}) into dual minimization problems, where $y$, $\bm{p}$, $z$, and $\bm{q}$ are dual design variables corresponding to Eqs.~\eqref{con1:inner_min1}, \eqref{con2:inner_min1}, \eqref{con1:inner_min2}, and \eqref{con2:inner_min2}, respectively. \hfill $\square$
\end{proof}

Although there are bilinear terms (i.e., $x_i x_j$) of binary design variables in Eqs.~\eqref{obj:robust2} and \eqref{con3:robust2}, they can be linearized using well-known reformulation techniques~\cite{williams2013model}. 
As a result, problem (Eqs.~\eqref{obj:robust2}--\eqref{con4:robust2}) can be reduced to a mixed-integer linear optimization problem.
Note also that Gurobi, which is the commercial optimization solver used in our experiments, is capable of directly solving problem~(Eqs.~\eqref{obj:robust2}--\eqref{con4:robust2}).

\section{Experiments}\label{sec:experiment}
In this section, we evaluate the effectiveness of our robust portfolio optimization method for recommendation through numerical experiments.
All experiments were performed on a Mac OS 14.4 computer equipped with an Apple M3 processor @ 4.05 GHz (8 cores) and 24 GB RAM. 

\subsection{Datasets}\label{subsec:datasets}
We used two publicly available datasets of user ratings, the MovieLens 100K\footnote[1]{\url{https://grouplens.org/datasets/movielens/100k/}} and Yahoo! R3\footnote[2]{\url{https://webscope.sandbox.yahoo.com/}} (Yahoo! music ratings for user selected and randomly selected songs, v.~1.0) datasets, commonly employed to evaluate recommender systems.

The MovieLens 100K dataset contains 100,000 ratings by 943 users on 1,682 movies on a scale of 1 to 5. 
Since some users provided only a few ratings, we selected 568 users who rated at least 50 movies. 
Each user's ratings were randomly split into training (60\%) and testing (40\%) sets, and the results were averaged over five random splits. 

The Yahoo! R3 dataset is pre-divided into training and testing sets, which contain 15,400 and 5,400 users, respectively. 
Each user rated some of the 1,000 songs on a scale of 1 to 5, and the testing set contained 10 ratings per user.
We extracted ratings from the 5,400 users included in both the training and testing sets.

\subsection{Experimental setup}\label{subsec:setup}
For rating prediction, we implemented the singular value decomposition method using the Python \texttt{Surprise} library\footnote[3]{\url{http://surpriselib.com/}}~\cite{hug2020surprise}. 
Here, we set the number of factors to 100, the learning rate to 0.01, and the regularization weight to 0.1.
The root mean square error (RMSE) for rating prediction was 0.92 and 1.42  on the MovieLens 100K and Yahoo! R3 testing datasets, respectively. 
This indicates the prediction accuracy was lower for the Yahoo! R3 dataset than for the MovieLens 100K dataset.

For each user $u \in U$, the candidate item set $I_{u}$ was composed of the items rated by the user in the corresponding testing set. 
The number of recommended items was set to $N=10$ for the MovieLens 100K dataset and $N=5$ for the Yahoo! R3 dataset.
We set $\alpha = 0.2$ as the hyperparameter for adjusting the mean--variance trade-off. 
We performed shrinkage estimation of rating covariance matrices as 
\[
\hat{\bm{\Sigma}} \coloneqq (\hat{\sigma}_{ij})_{(i,j) \in I \times I} = 0.5\bm{S} + 0.5\bm{F},
\]
where $\bm{S}$ is the sample covariance matrix, and $\bm{F}$ is the matrix-completion-based target matrix; see Yasumoto and Takano~\cite{yasumotomean} for details. 

We solved the portfolio optimization problem~(Eqs.~\eqref{obj:robust2}--\eqref{con4:robust2}) using the Gurobi Optimizer\footnote[4]{\url{https://www.gurobi.com/}}, with a maximum computation time of 3 seconds. 
We set the ranges of variation in the expectation and covariance of ratings as
\begin{align}
& \delta_{ui}^{(\mu)} = \frac{\sqrt{\hat{\sigma}_{ii}}}{\sqrt{\max\{1, n_i\}}} \quad ((u,i) \in U \times I), \quad 
\delta_{ij}^{(\sigma)} = \frac{0.2 \cdot \hat{\sigma}_{ij}}{\sqrt{\max\{1,n_{ij}\}}} \quad ((i,j) \in I \times I), \notag
\end{align}
where $n_i$ is the number of users who rated item $i \in I$, and $n_{ij}$ is the number of users who rated both items $i,j \in I$.

\subsection{Evaluation Metrics}

To evaluate both the accuracy and diversity of recommendations, we used the following evaluation metrics~\cite{shani2011evaluating}: 
\begin{itemize}
\item \textbf{F1 score}: average F1 score of recommendations across all users;  
\item \textbf{Gini coefficient}: Gini coefficient of the number of times each item was recommended. 
\end{itemize} 

The F1 score was employed to measure the recommendation accuracy based on the intersection of recommended and highly rated items. 
Here, highly rated items were defined as items with a rating of 4 or higher in the MovieLens 100K dataset and a rating of 3 or higher in the Yahoo! R3 dataset. 
The higher the F1 score, the more accurate the recommendations.

The Gini coefficient was employed to measure the inequality in the number of recommendations among items. 
The lower the Gini coefficient, the more diverse the recommendations.

\begin{figure}[t]
    \centering
    \subfigure[F1 score vs. $\Gamma^{(\mu)}$]{\includegraphics[scale=0.15]{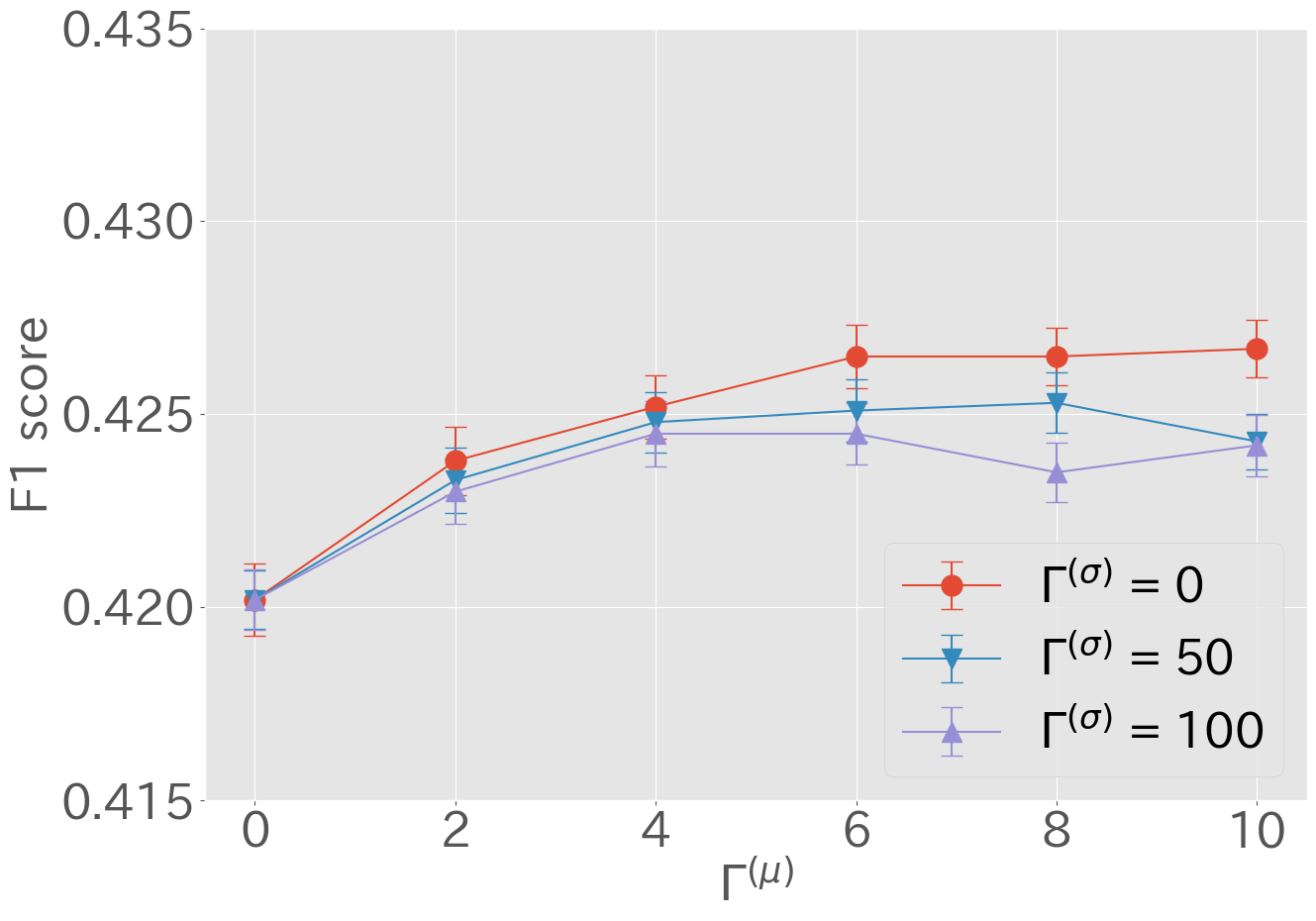}\label{fig:result_movie_f1_mu}}
    \subfigure[F1 score vs. $\Gamma^{(\sigma)}$]{\includegraphics[scale=0.15]{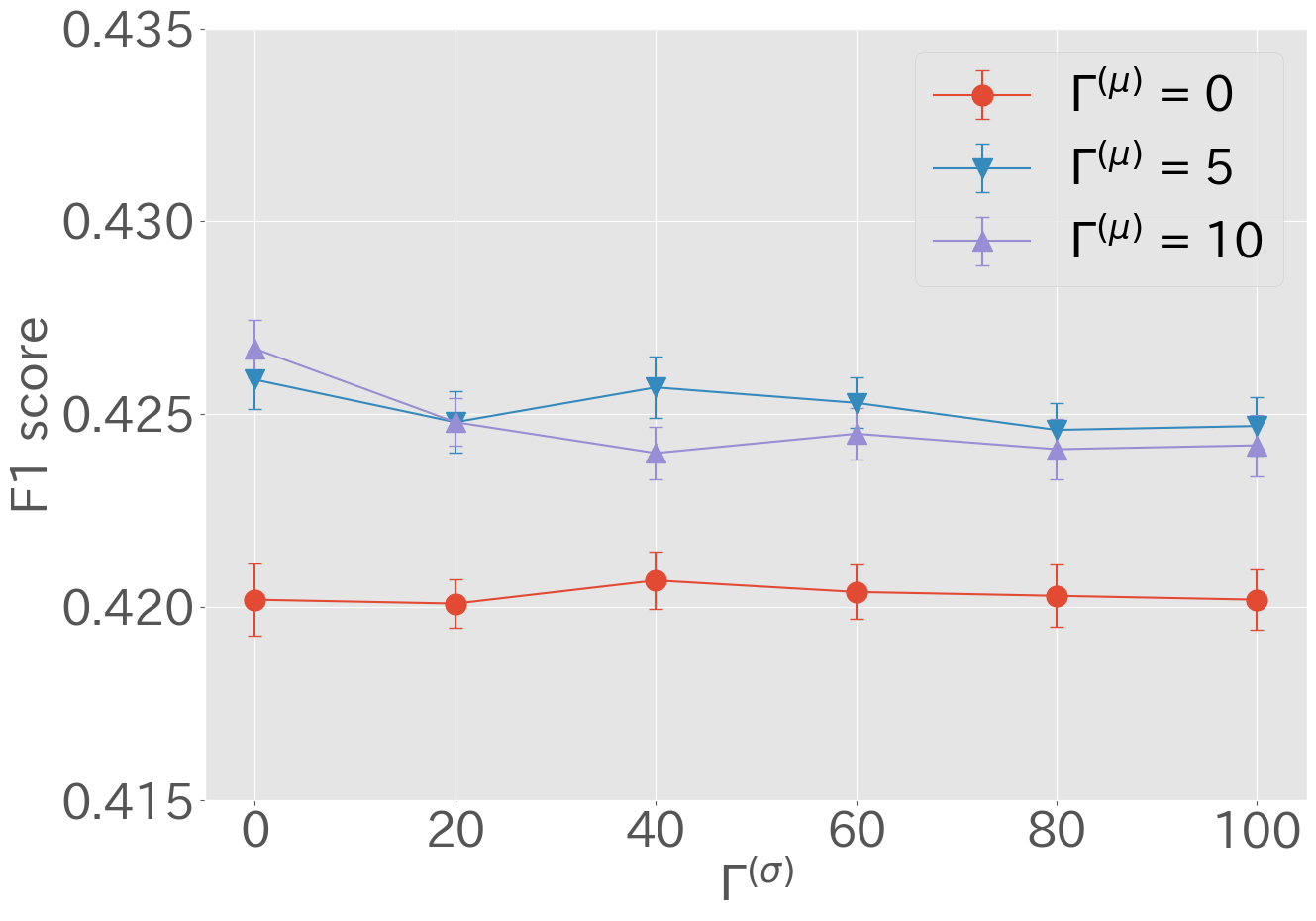}\label{fig:result_movie_f1_sigma}}
    \subfigure[Gini coefficient vs. $\Gamma^{(\mu)}$]{\includegraphics[scale=0.15]{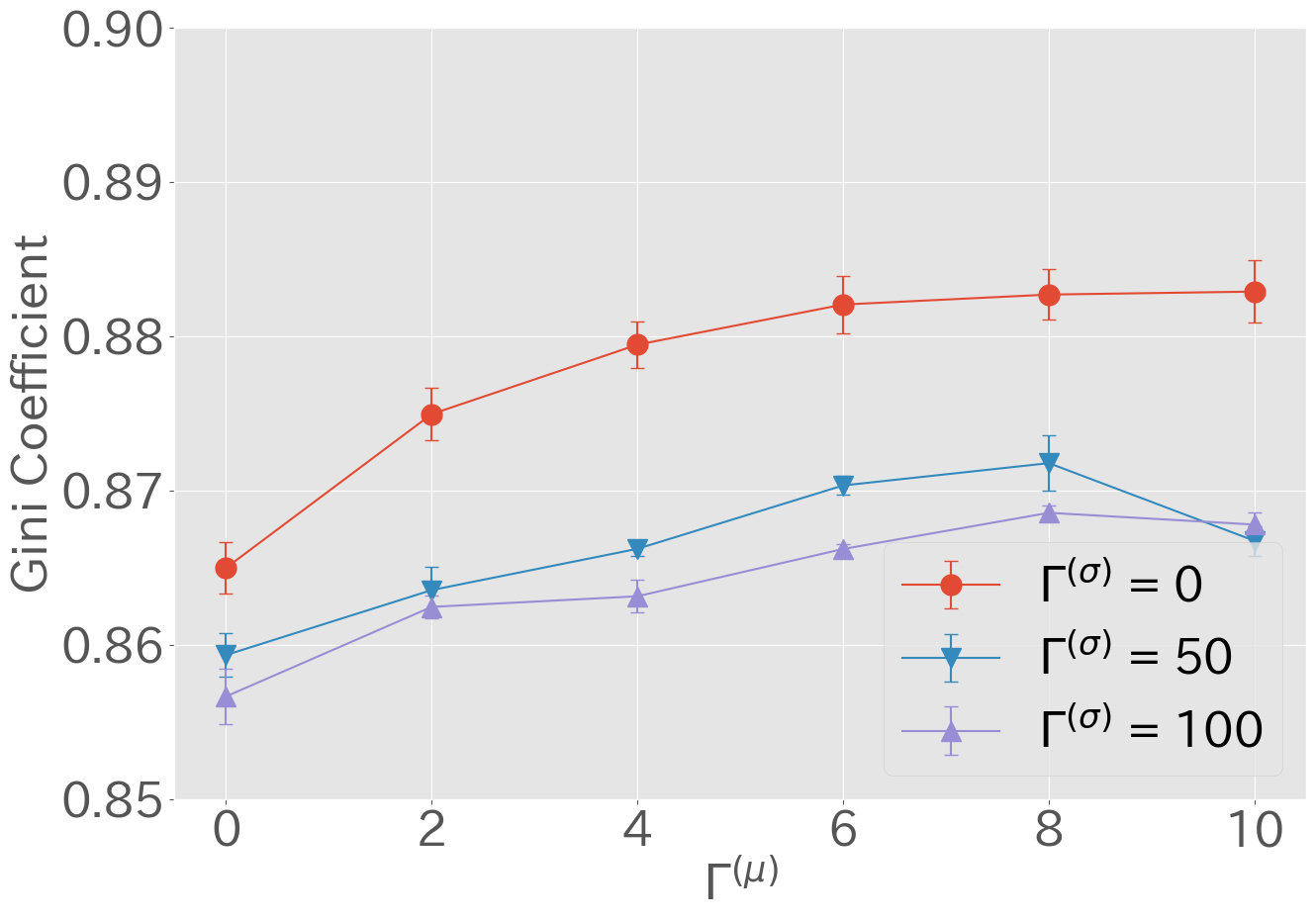}\label{fig:result_movie_gini_mu}}
    \subfigure[Gini coefficient vs. $\Gamma^{(\sigma)}$]{\includegraphics[scale=0.15]{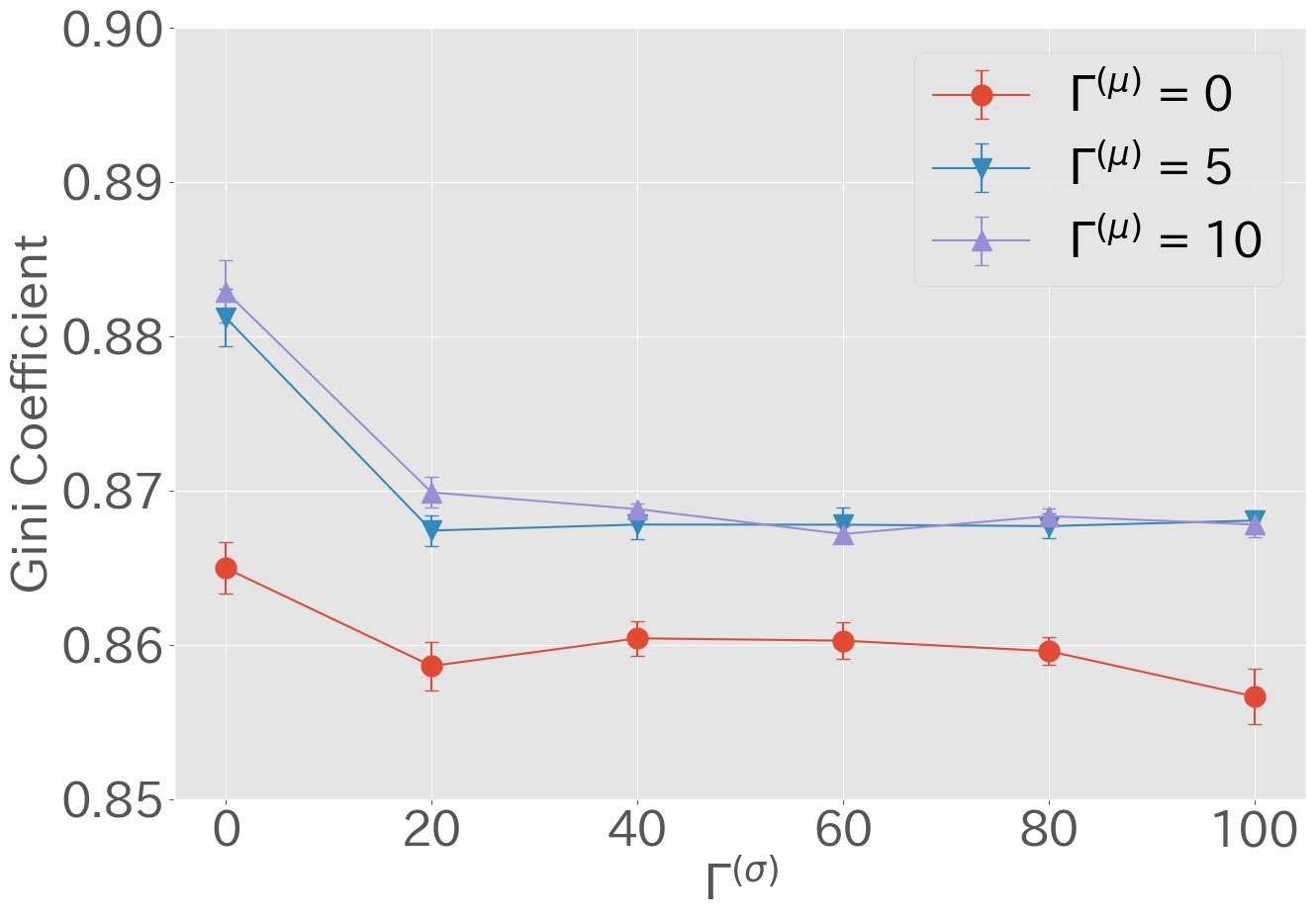}\label{fig:result_movie_gini_sigma}}
    \caption{Results for the MovieLens 100K dataset ($N=10$)}
    \label{fig:result_movie}
\end{figure}

\subsection{Results for the MovieLens 100K dataset}\label{subsec:results_mv}

Fig.~\ref{fig:result_movie} shows the F1 scores and Gini coefficients of recommendations ($N=10$) with various values of $\Gamma^{(\mu)}$ and $\Gamma^{(\sigma)}$ on the MovieLens 100K dataset.  
Recall that these results were averaged over five repetitions, with standard errors displayed as error bars in the figures. 

We first focus on the effect of the cardinality parameter $\Gamma^{(\mu)}$ of variation in the expected ratings. 
The F1 score improved as $\Gamma^{(\mu)}$ increased (Fig.~\ref{fig:result_movie_f1_mu}), which suggests that considering the uncertainty in expected ratings can upgrade the accuracy of recommendations.
Conversely, the Gini coefficient increased as $\Gamma^{(\mu)}$ increased (Fig.~\ref{fig:result_movie_gini_mu}), which implies that considering the uncertainty in expected ratings can decrease the diversity of recommendations among items.
This is likely because popular items are preferentially recommended to a wide range of users to avoid uncertainty in expected ratings.

We next move on to the effect of the cardinality parameter $\Gamma^{(\sigma)}$ of variation in the rating covariance. 
Increasing $\Gamma^{(\sigma)}$ had relatively small impacts on the F1 score (Fig.~\ref{fig:result_movie_f1_sigma}). 
In contrast, increasing $\Gamma^{(\sigma)}$ led to a reduction in the Gini coefficient (Fig.~\ref{fig:result_movie_gini_sigma}), which suggests that considering the uncertainty in rating covariance can enhance the diversity of recommendations among items.
This effect is due to the recommendation of item pairs with lower covariance, resulting in more diverse sets of recommended items.

These results on the MovieLens 100K dataset (Fig.~\ref{fig:result_movie}) demonstrate that our robust portfolio optimization method can significantly improve the accuracy or diversity of recommendations by appropriately tuning the cardinality parameters $\Gamma^{(\mu)}$ and $\Gamma^{(\sigma)}$, compared to not considering the uncertainty in the expectation or covariance of ratings (i.e., $\Gamma^{(\mu)}=\Gamma^{(\sigma)}=0$).

\begin{figure}[t]
    \centering
    \subfigure[F1 score vs. $\Gamma^{(\mu)}$]{\includegraphics[scale=0.15]{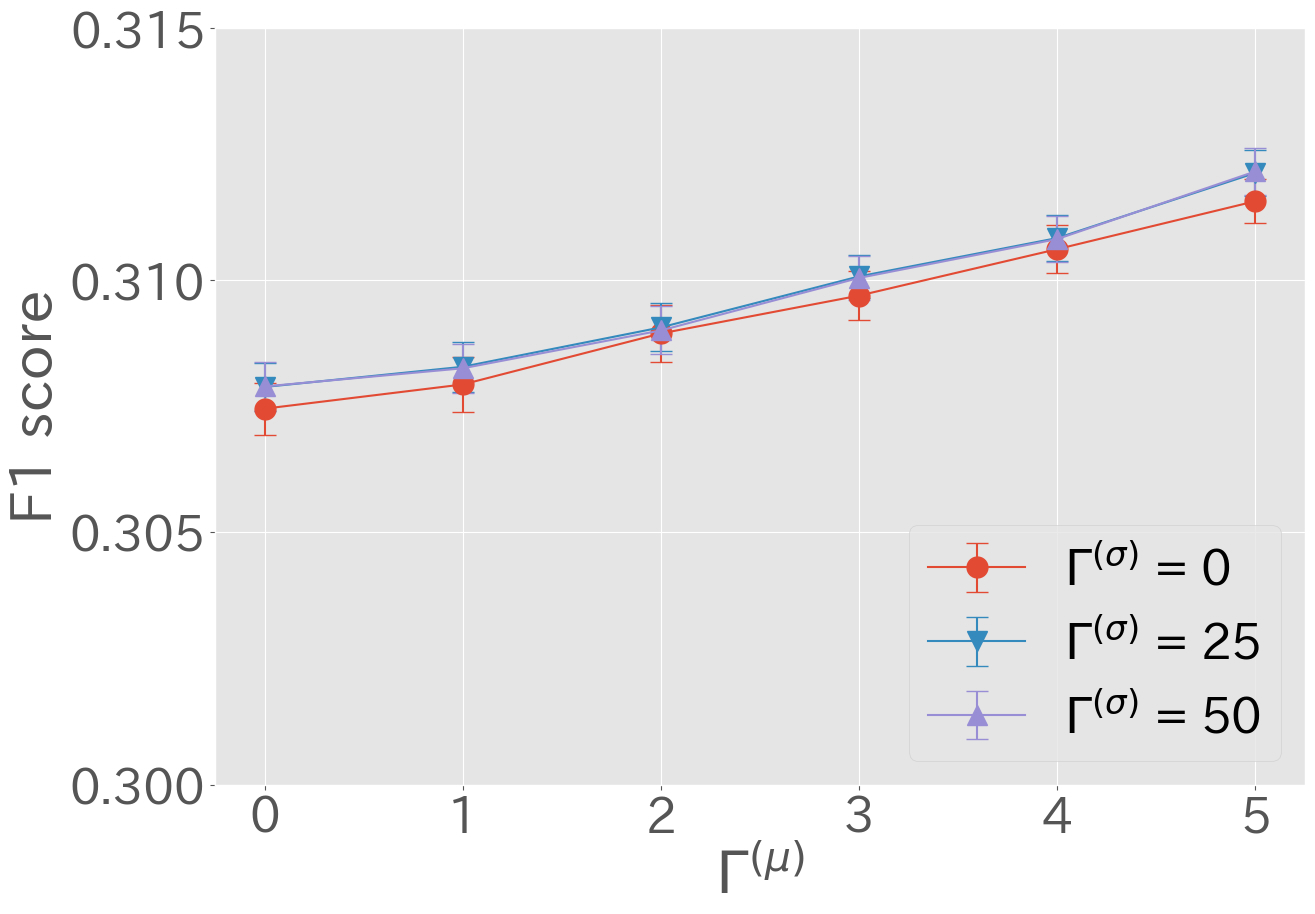}\label{fig:result_r3_f1_mu}}
    \subfigure[F1 score vs. $\Gamma^{(\sigma)}$]{\includegraphics[scale=0.15]{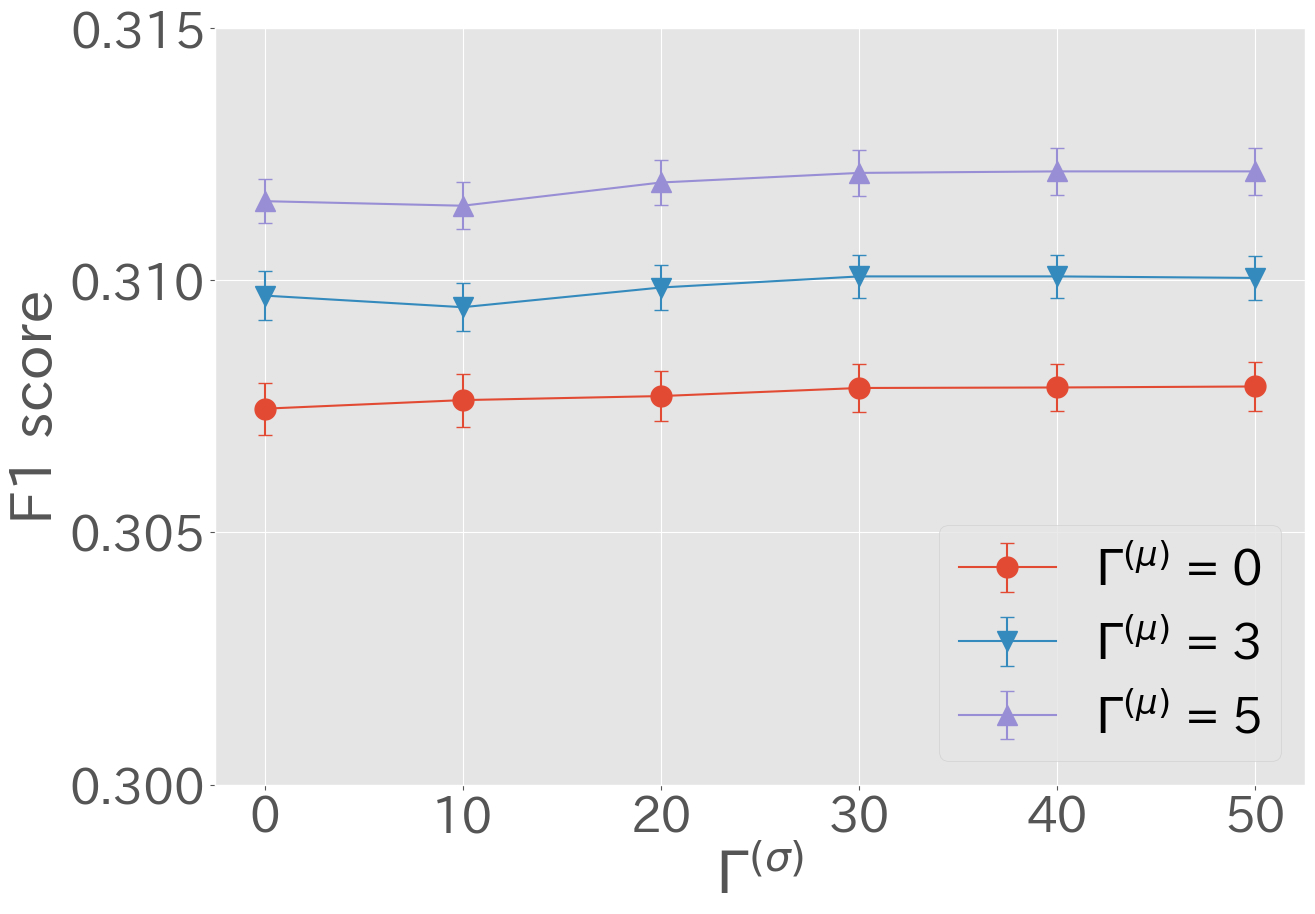}\label{fig:result_r3_f1_sigma}}
    \subfigure[Gini coefficient vs. $\Gamma^{(\mu)}$]{\includegraphics[scale=0.15]{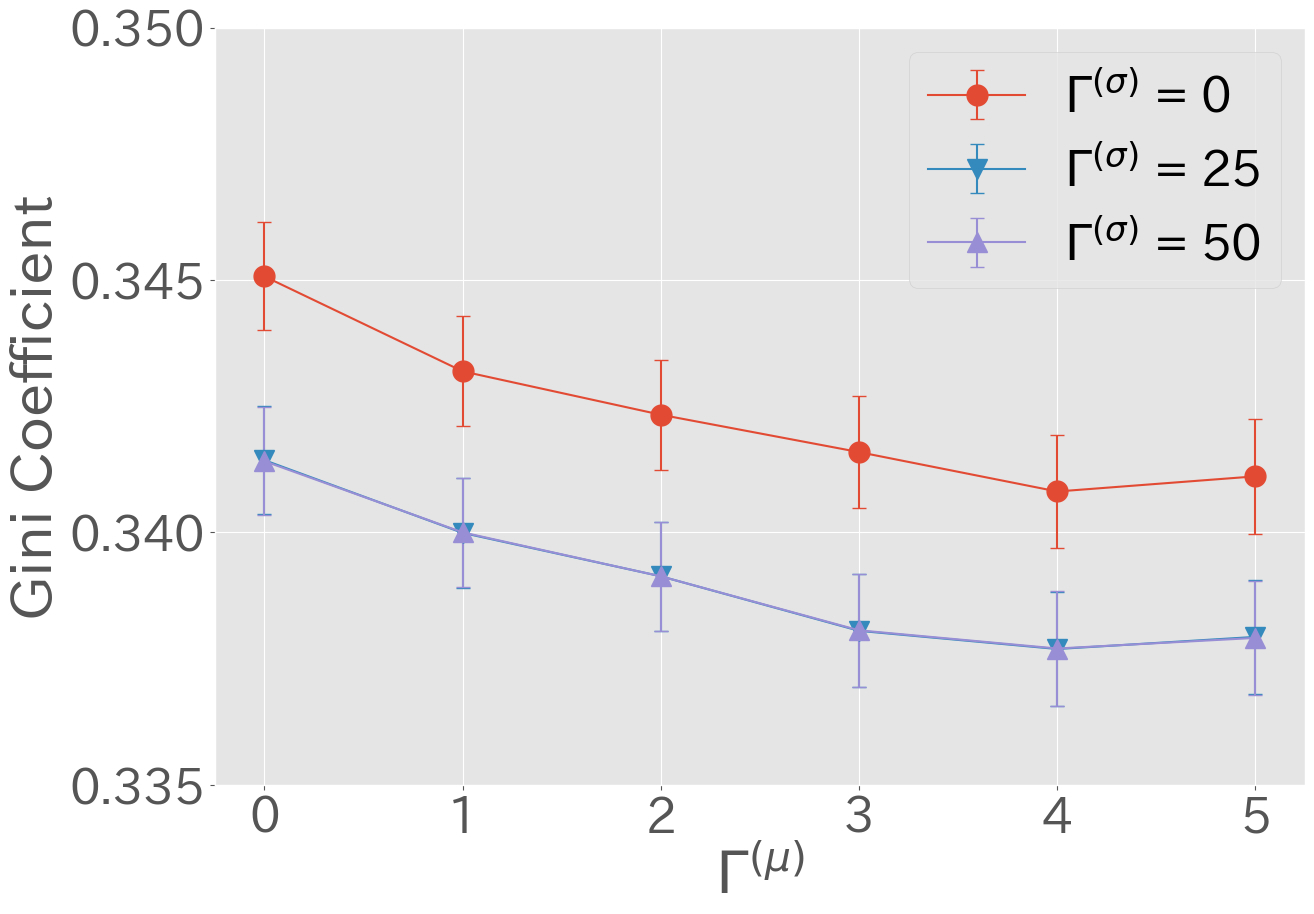}\label{fig:result_r3_gini_mu}}
    \subfigure[Gini coefficient vs. $\Gamma^{(\sigma)}$]{\includegraphics[scale=0.15]{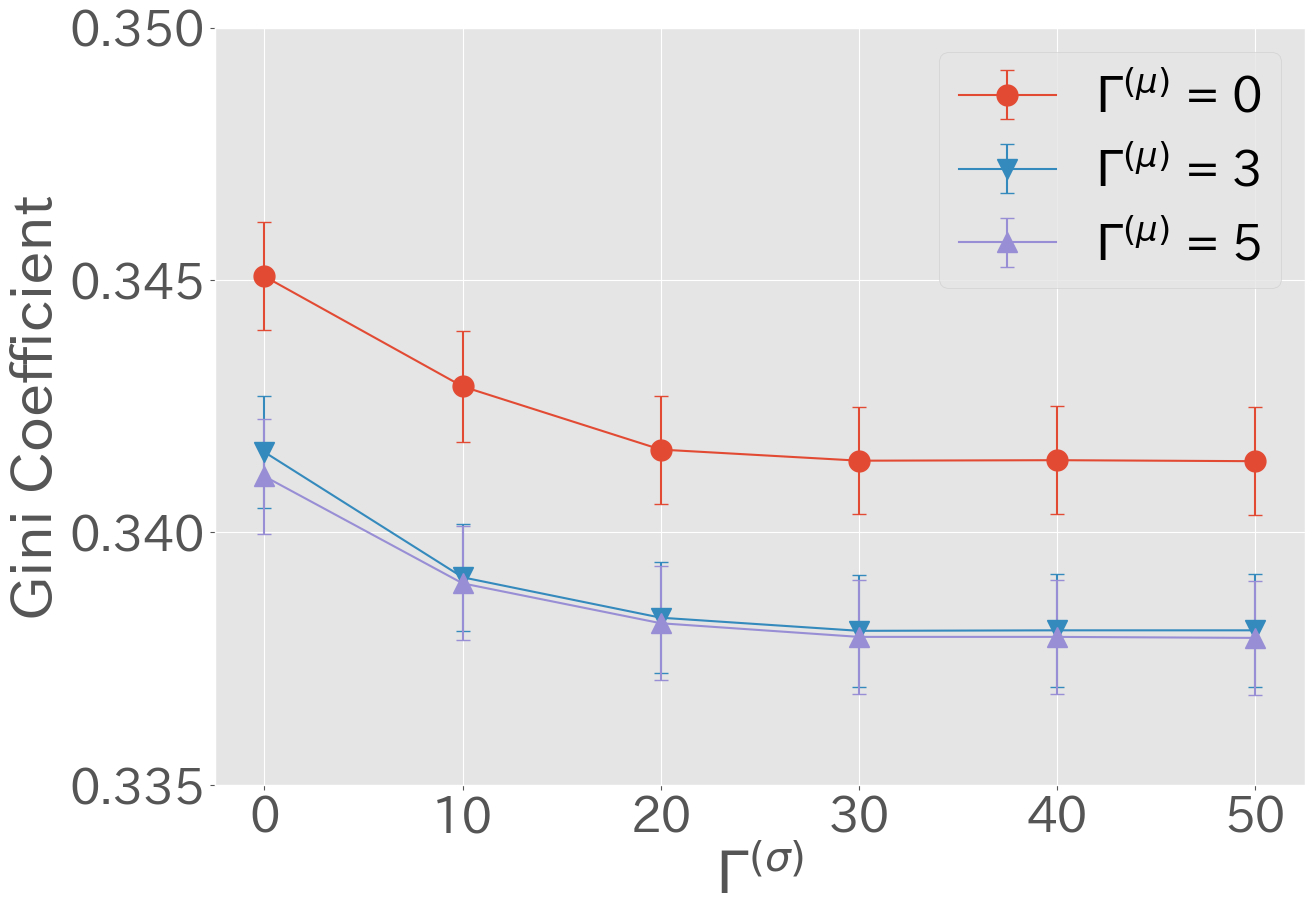}\label{fig:result_r3_gini_sigma}}
    \caption{Results for the Yahoo! R3 dataset ($N=5$)}
    \label{fig:result_r3}
\end{figure}

\subsection{Results for the Yahoo! R3 dataset}\label{subsec:results_r3}

Fig.~\ref{fig:result_r3} shows the F1 scores and Gini coefficients of recommendations ($N=5$) with various values of $\Gamma^{(\mu)}$ and $\Gamma^{(\sigma)}$ on the Yahoo! R3 dataset. 
Note that the error bars displayed in the figures represent standard errors estimated from 10 bootstrap runs.

We first focus on the effect of the cardinality parameter $\Gamma^{(\mu)}$ of variation in the expected ratings. 
The F1 score improved as $\Gamma^{(\mu)}$ increased (Fig.~\ref{fig:result_r3_f1_mu}) as well as the MovieLens 100K dataset (Fig.~\ref{fig:result_movie_f1_mu}). 
Moreover, in contrast to Fig.~\ref{fig:result_movie_gini_mu} on the MovieLens 100K dataset, the Gini coefficient decreased as $\Gamma^{(\mu)}$ increased (Fig.~\ref{fig:result_r3_gini_mu}). 
This indicates that considering the uncertainty in expected ratings can significantly improve both the accuracy and diversity of recommendations.

We next move on to the effect of the cardinality parameter $\Gamma^{(\sigma)}$ of variation in the rating covariance. 
Increasing $\Gamma^{(\sigma)}$ had relatively small impacts on the F1 score (Fig.~\ref{fig:result_r3_f1_sigma}) as well as the MovieLens 100K dataset (Fig.~\ref{fig:result_movie_f1_sigma}). 
In contrast, increasing $\Gamma^{(\sigma)}$ led to a reduction in the Gini coefficient (Fig.~\ref{fig:result_r3_gini_sigma}) as well as the MovieLens 100K dataset (Fig.~\ref{fig:result_movie_gini_sigma}). 
This suggests that considering the uncertainty in rating covariance can significantly enhance the diversity of recommendations among items.

These results on the Yahoo! R3 dataset (Fig.~\ref{fig:result_r3}) demonstrate that our robust portfolio optimization method can improve recommendation diversity without reducing recommendation accuracy, compared to not considering the uncertainty in the expectation or covariance of ratings (i.e., $\Gamma^{(\mu)}=\Gamma^{(\sigma)}=0$).
This positive result can be attributed to the fact that the accuracy of rating prediction was lower on the Yahoo! R3 dataset (RMSE: 1.42) than on the MovieLens 100K dataset (RMSE: 0.92), as mentioned in Section~\ref{subsec:setup}.

\subsection{Discussion}\label{subsec:discussions}
Yasumoto and Takano~\cite{yasumotomean} demonstrated that increasing $\alpha$ in Eq.~\eqref{eq:portfolio_obj} improves recommendation diversity.
This trend is similar to our results when increasing $\Gamma^{(\sigma)}$. 
This similarity likely arises because both $\alpha$ and $\Gamma^{(\sigma)}$ encourage the inclusion of item pairs with low covariance in the recommendation list.

The previous studies~\cite{wang2009mean,wang2009portfolio} in information retrieval demonstrated that mean--variance portfolio analysis increases diversity of ranked documents and improves retrieval performance. 
Thus, our robust optimization model is also potentially effective in ranking documents. 
Moreover, our recommendation performance can be further improved by properly tuning $\alpha$. 

Finally, we focus on experimental results of robust optimization models in other domains, such as personalized pricing~\cite{ikeda2024robust} and coupon allocation~\cite{uehara2024robust}.
These studies demonstrated that the model performance and its standard error tended to improve as the degree of considered uncertainty raised. 
A similar trend was observed in our experiments, where the F1 scores and their standard errors improved with increasing $\Gamma^{(\mu)}$ on both real-world datasets.

\section{Conclusion}\label{sec:conclution}
We proposed a robust portfolio optimization model for recommendation that copes with the uncertainty of estimated statistics. 
Specifically, we introduced the cardinality-based uncertainty sets~\cite{bertsimas2004price} of the expectation and covariance of ratings in the mean--variance portfolio optimization model for recommender systems. 
Our robust portfolio optimization model can be reduced to a mixed-integer linear optimization problem, which can be solved exactly using mathematical optimization solvers. 

We conducted numerical experiments using two publicly available rating datasets.
For the MovieLens 100K dataset, our method improved the accuracy or diversity of recommendations by appropriately tuning the cardinality parameters. 
For the Yahoo! R3 dataset, our method improved the recommendation diversity without reducing the recommendation accuracy.
Importantly, our method has the potential to improve the
recommendation quality of various rating prediction algorithms (e.g., collaborative filtering \cite{su2009survey}, matrix factorization \cite{koren2009matrix}, and deep learning techniques~\cite{gao2022graph,zhang2019deep}). 

Solving the mixed-integer optimization problem (Eqs.~\eqref{obj:robust2}--\eqref{con4:robust2}) becomes challenging as the numbers of users and items increase. 
A future direction of study will be to apply high-performance algorithms for solving cardinality-constrained portfolio optimization problems~\cite{chang2000heuristics,kobayashi2023cardinality}. 
Other research directions for our method include combining it with price optimization~\cite{den2015dynamic,ikeda2023prescriptive}, developing privacy preservation mechanisms for it~\cite{himeur2022latest,yanagi2024privacy}, and extending it to multiperiod/dynamic portfolio optimization~\cite{boyd2017multi,takano2023dynamic}. 

%
%
%
%
\bibliographystyle{splncs04}
\bibliography{cite.bib} 

\end{document}